\documentclass[conference]{IEEEtran}

\usepackage[dvips]{graphicx}
\usepackage{amsmath,amssymb}
\usepackage{algorithm}
\usepackage{algorithmic}

\usepackage{amsthm}
\usepackage{flushend}
\usepackage[english]{babel}
\theoremstyle{remark} 
\newtheorem{oss}{Remark}
\theoremstyle{plain}
\newtheorem{theorem}{Theorem}
\newtheorem{prop}{Proposition}
\newtheorem{lem}{Lemma}
\newtheorem{cor}{Corollary}
\newtheorem {con}{Conjecture}

\begin{document}
\title{On Some Properties of Quadratic APN Functions of a Special Form}
\date{}%date stay empty

\author{
\IEEEauthorblockN{Irene Villa}
\IEEEauthorblockA{Department of Informatics\\ University of Bergen \\ Bergen, Norway \\ Irene.Villa@uib.no\\}
}
\maketitle

\begin{abstract}
In a recent paper \cite{BCL09.2}, it is shown that functions of the form $L_1(x^3)+L_2(x^9)$, where $L_1$ and $L_2$ are linear,
are a good source for construction of new infinite families of APN functions.
In the present work we study necessary and sufficient conditions for such functions to be APN. 
\end{abstract}

\section{Introduction}
For given positive integers $n$ and $m$, a function $F$ from the finite field with $2^n$ elements to the finite field with $2^m$ elements is called a vectorial Boolean function or an ($n,m$)-function, and in the case when $m$=1 it is simply called a Boolean function.
Boolean functions are among the most fundamental objects investigated in pure and applied mathematics and computer science. Boolean function theory is an important tool for solving problems of analysis and synthesis of discrete devices which transform and process information.
The primary motivation for studying Boolean functions comes from cryptography.
In modern society, exchange and storage of information in an efficient, reliable and secure manner is of fundamental importance.
Cryptographic primitives are used to protect information against eavesdropping, unauthorized changes and other misuse.
In the case of symmetric cryptography ciphers are designed by appropriate composition of nonlinear Boolean functions.
For example, the security of block ciphers depends on S-boxes which are ($n,m$)-functions. For most of cryptographic attacks on block ciphers there are certain properties of functions which measure the resistance of the S-box to these attacks. 
The differential attack introduced by Biham and Shamir is one of the most efficient cryptanalysis tools for block ciphers. It is based on the study of how differences in an input can affect the resulting difference at the output.
When $n=m$ the functions that contribute an optimal resistance against differential attack are called \emph{Almost Perfect Nonlinear} (APN). Such APN function $F(x)$ are characterized by having at most two solution for every equation $F(x+a)-F(x)=b$, where $a$ and $b$ are general elements of the field and $a$ is not null.\\
The role of APN functions is not just related to cryptography.
In coding theory APN functions define binary error correcting codes optimal in a certain sense.
In projective geometry quadratic APN functions define dual hyperovals.
Recent advances in APN functions have made a prominent impact on the theory of commutative semifields.\\
For these reasons many works were focused on the studying and the construction of such optimal functions. 

Let assume $n$ be a positive integer and $\mathbb{F}_{2^n}$ the finite field with $2^n$ elements.
If $n$ is an even number, then we have that $3 | (2^n-1)$ and we denote with $k$ the integer value $\frac{2^n-1}{3}$.\\
A function $F$ from $\mathbb{F}_{2^n}$ to $\mathbb{F}_{2^n}$ admits a unique representation, called \emph{Univariate Polynomial Representation}, over $\mathbb{F}_{2^n}$  of degree at most $2^n-1$: $$F(x)=\sum_{j=0}^{2^n-1}\delta_jx^j, \mbox{ with }\delta_j\in\mathbb{F}_{2^n}.$$
For every integer $j$ consider its binary expansion $\sum_{s=0}^{n-1}j_s2^s$ and denote with $w_2(j)$ the number of nonzero coefficients (i.e. $\sum_{s=0}^{n-1}j_s$). 
The \emph{algebraic degree} of the function $F$ is the $\max_{j=0,\ldots,2^n-1/ \delta_j\neq0}w_2(j)$.
Functions of algebraic degree 1 are called \emph{affine} and of degree 2 \emph{quadratic}.
\emph{Linear functions} are affine functions without the constant term and they can be represented as $L(x)=\sum_{j=0}^{n-1}\gamma_jx^{2^j}$.
A known example of a linear function defined over any dimension $n$ is the \emph{Trace function}
$\textit{Tr}(x)=\textit{Tr}_n(x)=\sum_{i=0}^{n-1}x^{2^i},$
In particular the trace is a Boolean function, i.e. $\textit{Tr} : \mathbb{F}_{2^n} \rightarrow \mathbb{F}_2$.
For $m$ positive divisor of $n$ we use the notation $\textit{Tr}^m(x)=\sum_{i=0}^{n/m-1}x^{2^{im}}$.

Given a function $F$ we define its $\lambda$-component as the Boolean function $f_\lambda:\mathbb{F}_{2^n} \rightarrow \mathbb{F}_2$ with $f_\lambda(x)=\textit{Tr}(\lambda\cdot F(x))$, for $\lambda\in\mathbb{F}_{2^n}$.
For a Boolean function $f$ we define the Walsh transformation as $$\hat{f}_\chi (u) =\sum_{x\in\mathbb{F}_{2^n}}(-1)^{f(x)+\textit{Tr}(ux)},$$ with $u\in\mathbb{F}_{2^n}$. With \emph{Walsh Spectrum} we refer to the set of all possible values of the Walsh transformation. 
With the symbol $\mathcal{F}(f)$ we indicate the Walsh transformation valued in 0, $$\mathcal{F}(f)=\sum_{x\in\mathbb{F}_{2^n}}(-1)^{f(x)}= 2^n - 2\cdot\mbox{wt}(f),$$ where wt$(f)$ is the Hamming weight of $f$ (i.e. the cardinality of the set $ \{ x \in\mathbb{F}_{2^n} : f(x) = 1 \} $).
Therefore we have that a Boolean function $f$ is balanced (wt$(f)=2^{n-1}$) if and only if  $\mathcal{F}(f)=0$.
A Boolean function $f$ is called \emph{bent} if its Walsh spectrum corresponds to the set $\{ \pm 2^{n/2} \}$. Therefore such function can exist only for even values of $n$.
 Moreover, we have that $f$ is bent if and only if, for every $a\in\mathbb{F}_{2^n}^*$, the function $D_a f(x)=f(x+a)+f(x)$ is balanced.
 
For every nonzero element $a\in\mathbb{F}_{2^n}^*$ the \emph{derivative} of $F$ in the direction of $a$ is the function  $D_a F(x) = F(x+a)+F(x)$.
The function $F$ is called \emph{almost perfect nonlinear} (APN) if for every $a\neq0$ and every $b$ in $\mathbb{F}_{2^n}$, the equation $D_a F(x)=b$ admits at most 2 solutions.
Used as S-Boxes in block ciphers, APN functions are useful since they oppose an optimal resistance against differential cryptanalysis.
 
 The APN property is invariant under the action of some transformations of functions.
 \begin{itemize}
 \item Given $A_1,A_2$ affine permutations and $A$ an affine function, if $F$ is APN then also $G=A_1 \circ F \circ A_2 +A$ is APN; in this case the functions are called \emph{extended affine equivalent} (EA-equivalent).
 \item Two functions $F$ and $G$ are \emph{CCZ-equivalent} if there exists an affine permutation $\mathcal{L}$ of $\mathbb{F}_{2^n}^2$ such that $\mathcal{L}(\Gamma_F) = \Gamma_G$, where $\Gamma_F$ is the graph of the function $F$, $\{(x,F(x)) : x \in \mathbb{F}_{2^n} \}$. 
 Also CCZ-equivalence preserve the APN property. Moreover,  we have that EA-equivalence is a particular case of CCZ-equivalence.
 \end{itemize}
Many different works have been focused on finding and constructing new families of APN functions.
Table \ref{powerAPN} gives us all known values for exponents $d$ such that the function $x^d$, defined over $\mathbb{F}_{2^n}$, is APN.

\begin{table}[b]
\centering
\label{powerAPN}
\caption{Known APN power functions $x^d$ over $\mathbb{F}_{2^n}$}\label{tab1}\vspace{4pt}
\renewcommand{\arraystretch}{1.2}
\begin{tabular}{|c|c|c|c|c|}
		\hline
 		Functions & Exponents $d$ & Conditions & Degree & Proven  \\
 		\hline
		Golden & $2^i+1$ & gcd($i,n$)=1 & 2 & \cite{86, 130}\\
		\hline
		Kasami & $2^{2i}-2^i+1$ & gcd($i,n$)=1 & $i$+1 & \cite{102,103}\\
		\hline
		Welch & $2^t+3$ & $n=2t+1$ & 3 & \cite{75}\\
		\hline
		Niho & $2^t+2^{\frac{t}{2}}-1$, $t$ even & $n=2t+1$ & $\frac{t+2}{2}$ & \cite{74}\\
		  & $2^t+2^{\frac{3t+1}{2}}-1$, $t$ odd &  & $t$+1 & \\
		\hline
		Inverse & $2^{2t}-1$ & $n=2t+1$ & $n-1$ & \cite{7, 130}\\
		\hline
		Dobbertin & $2^{4i}+2^{3i}+2^{2i}+2^i-1$ & $n=5i$ & $i+3$ &  \cite{76}\\
		\hline
\end{tabular}
\end{table}

Since EA-equivalence preserves the algebraic degree of a function and, in general, the functions listed in Table \ref{powerAPN} have different algebraic degrees, it is easy to verify that these APN functions are EA-inequivalent.
Instead the algebraic degree is not an invariant for CCZ-equivalence.
But also for this case it was possible to prove some inequalities.
In \cite{33} it is shown that two different Gold functions $x^{2^i+1}$ and $x^{2^j+1}$, where $1\leq i<j\leq n/2$, are CCZ-inequivalent and that in general the Gold functions are CCZ-inequivalent to the Welch and to any Kasami functions.
Moreover, the inverse and Dobbertin functions are not CCZ-equivalent to each other and to all other known APN power functions, \cite{33}.
For all the other cases the problem is still open.\\
Before the work in \cite{34} the only known APN functions were EA-equivalent to power functions and it was supposed that all APN functions are EA-equivalent to power functions.
In \cite{34} it is showed the existence of classes of APN mappings EA-inequivalent to power functions. 
Such functions were constructed by applying CCZ-equivalence to the Gold APN mappings. 
In \cite{83} we can find the first examples of APN function CCZ-inequivalent to power functions.
The first infinite families of such APN polynomial can be found in \cite{34}.
In Table II these functions are listed. They are all quadratic functions.

\begin{table*}[h]
\centering
\label{polAPN}
\caption{Known classes of quadratic APN polynomial over $\mathbb{F}_{2^n}$ CCZ-inequivalent to power functions}\label{tab1}\vspace{4pt}
\renewcommand{\arraystretch}{1.4}
\begin{tabular}{|c|c|c| }
		\hline
		Functions &  Conditions &  Proven \\
 		\hline 
		& $n=pk$,  gcd$(k,3)$=gcd($s,3k$)=1, & \\
		$x^{2^s+1}+\alpha^{2^k-1}x^{2^{ik}+2^{mk+s}}$ & $p\in\{3,4\}$, $i=sk$ mod $p$, $m=p-i$,   & \cite{30} \\
		& $n\geq12$, $\alpha$ primitive in $\mathbb{F}_{2^n}^*$   & \\
		\hline
		& $q=2^m$, $n=2m$, gcd($i,m$)=1,   & \\
		$x^{2^{2i}+2^i}+bx^{q+1}+cx^{q(2^{2i}+2^i)}$ & gcd($2^i+1,q+1)\neq1$, $cb^q+b\neq0$,& \cite{26} \\
		& $c\not\in\{\lambda^{(2^i+1)(q-1)}, \lambda\in\mathbb{F}_{2^n}\}$, $c^{q+1}=1$ & \\
		\hline
		& $q=2^m$, $n=2m$, gcd($i,m$)=1,   & \\
		$x(x^{2^i}+x^q+cx^{2^iq}) $ & $c\in\mathbb{F}_{2^n}$, $s\in\mathbb{F}_{2^n}\smallsetminus\mathbb{F}_q$, & \cite{26}\\
		$+x^{2^i}(c^qx^q+sx^{2^iq})+x^{(2^i+1)q}$ & $X^{2^i+1}+cX^{2^i}+c^qX+1$   & \\
		& is irreducible over $\mathbb{F}_{2^n}\}$   & \\
		\hline
		$x^3+a^{-1}\textit{Tr}_n(a^3x^9)$& $a\neq0$ & \cite{BCL09}\\
		\hline
		$x^3+a^{-1}\textit{Tr}_n^3(a^3x^9+a^6x^{18})$& $3|n$, $a\neq0$ & \cite{BCL09.2} \\
		\hline
		$x^3+a^{-1}\textit{Tr}_n^3(a^6x^{18}+a^12x^36)$& $3|n$, $a\neq0$ & \cite{BCL09.2} \\
		\hline
		& $n=3k$, gcd$(k,3)$=gcd($s,3k$)=1,   & \\
		$ux^{2^s+1}+u^{2^k}x^{2^{-k}+2^{k+s}}+$ &  $v,w \in \mathbb{F}_{2^k}$, $vw\neq1$,  & \cite{14} \\
		$vx^{2^{-k}+1}+wu^{2^k+1}x^{2^s+2^{k+s}}$ & $3|(k+s)$ $u$ primitive in $\mathbb{F}_{2^n}^*$   & \\
		\hline
		& $n=2k$, gcd$(s,k)$=1, $s,k$ odd,   & \\
		$\alpha x^{2^s+1}+\alpha^{2^k}x^{2^{k+s}+2^{k}}+$ & $\beta\not\in\mathbb{F}_{2^k}$, $\gamma_i\in\mathbb{F}_{2^k}$, & \cite{13, 14}\\
		$\beta x^{2^{k}+1}+\sum_{i+1}^{k-1}\gamma_ix^{2^{k+i}+2^i}$ &  $\alpha$ not a cube   & \\
		\hline
\end{tabular}
\end{table*}

In this work we focus on functions of the form 
\begin{equation} \label{F'}
F^\prime(x) = F(x^3) = L_1(x^3)+L_2(x^9),
\end{equation}
where $L_1$ and $L_2$ are linear functions.
From now on, we will refer to $L_1$ and $L_2$ as to the linear functions
\begin{equation} \label{L}
L_1(x)=\sum_{i=0}^{n-1}b_ix^{2^i} \mbox{ and } L_2(x)=\sum_{i=0}^{n-1}c_ix^{2^i},
\end{equation}
 with $b_i,c_i\in\mathbb{F}_{2^n}$.
In particular we want to study conditions on $L_1$ and $L_2$ such that $F^\prime$ is APN.

Some results are already been given in different papers.
In \cite{BCL09} the function $x^3+\textit{Tr}(x^9)$ is proved to be APN for any dimension $n$. 
Moreover, for $n\geq 7$ it is proved to be CCZ-inequivalent to the Gold functions, to the inverse and Dobbertin functions and EA-inequivalent to power functions.
For a quadratic APN function $F:\mathbb{F}_{2^n} \rightarrow \mathbb{F}_{2^n}$ and a quadratic boolean function $f:\mathbb{F}_{2^n} \rightarrow \mathbb{F}_2$, under some conditions it is proved that the function $F(x)+f(x)$ is APN.
In particular these conditions are that for every nonzero $a\in\mathbb{F}_{2^n}$ there must exist a linear Boolean function $l_a$ satisfying:
\begin{enumerate}
\item $\varphi_f(x,a)=l_a(\varphi_F(x,a))$,
\item if $\varphi_F(x,a)=1$ for some $x\in\mathbb{F}_{2^n}$ then $l_a(1)=0$,
\end{enumerate}
where $\varphi_\chi(x,a)=\chi(x)+\chi(x+a)+\chi(a)+\chi(0)$.\\
A similar theorem is proved when $f:\mathbb{F}_{2^n}\rightarrow\mathbb{F}_{2^m}$, where $m$ is a divisor of $n$.
Due to this result the following functions, defined over $\mathbb{F}_{2^{2m}}$ where $m$ is an even positive integer, are APN:
\begin{itemize}
\item $x^3+\textit{Tr}_{n/m}(x^{2^m+2})=x^3+x^{2^m+2}+x^{2^{m+1}+1}$,
\item $x^3+(\textit{Tr}_{n/m}(x^{2^m+2}))^3$.
\end{itemize}
 When $F$ is a Gold function, all possible APN mappings $F(x)+f(x)$, where $f$ is a Boolean function, are computed until dimension 15.
The only possibilities, different from $x^3 + \textit{Tr}(x^9)$, are for $n=5$ the function $x^5+\textit{Tr}(x^3)$ (CCZ-equivalent to Gold functions) and for $n=8$ the function $x^9+\textit{Tr}(x^3)$ (CCZ-inequivalent to power functions and to $x^3+\textit{Tr}(x^9)$).

In \cite{BCL09.2} the function $x^3+\textit{Tr}(x^9)$ has been generalized to form (\ref{F'}). It has been proved that for $n$ even a sufficient condition is $L_1(x)+L_2(x^3)$ being a permutation over $\mathbb{F}_{2^n}$.
In the odd dimension a weaker conditions lead to an APN mapping.\\
Moreover from the fact that by applying a linear function $l(x)=ax+b$, with $a\in\mathbb{F}_{2^n}^*$ and $b\in\mathbb{F}_{2^n}$, to a permutation we obtain another one, a simply but useful statement has been proved.
In particular it is stated   that for $n$ even, $L$ a linear function over $\mathbb{F}_{2^n}$, $a\in\mathbb{F}_{2^n}^*$ and $b\in\mathbb{F}_{2^n}$ if $x+L(x^3)$ is  a permutation over $\mathbb{F}_{2^n}$, then
the function 
\begin{equation}\label{equat}
ax^3+L(a^3x^9+a^2bx^6+ab^2x^3)
\end{equation}
is APN over $\mathbb{F}_{2^n}$.\\
The statement just mentioned gives new examples of APN functions in even dimensions.
The following infinite families of function are proved to be APN also in odd dimensions:
\begin{enumerate}
\item $x^3 + a^{-1}\textit{Tr}(a^3x^9)$, with $a\in\mathbb{F}_{2^n}^*$ and any positive $n$;
\item $x^3+a^{-1}\textit{Tr}^3(a^6x^{18}+a^{12}x^{36})$, with $a\in\mathbb{F}_{2^n}^*$ and $n$ divisible by 3;
\item $x^3+a^{-1}\textit{Tr}^3(a^3x^9+a^6x^{18})$, with $a\in\mathbb{F}_{2^n}^*$ and $n$ divisible by 3.
\end{enumerate}
In \cite{BCL09.2} some conditions for constructing permutations, and consequentially for constructing APN functions, are given.
\begin{itemize}
\item For a general positive $n$ and a linear function $L$ over $\mathbb{F}_{2^n}$ if for every $u\in\mathbb{F}_{2^n}$ such that $L(u)\neq0$, the condition 
$$\textit{Tr}_n\bigg(\frac{u}{(L(u))^3}\bigg)=
\begin{cases}
0 & \mbox{ if $n$ is odd}\\
1 & \mbox{ if $n$ is even}
\end{cases}
$$
is satisfied, then the function $x+L(x^3)$ is a permutation.
\item For $n$ even integer and $L$ linear function over $\mathbb{F}_{2^n}$ the function $x+L(x^3)$ is a permutation of $\mathbb{F}_{2^n}$ if and only if for every $b\in\mathbb{F}_{2^n}^*$ such that $L^*(b)\neq0$ there exists an element $\gamma\in\mathbb{F}_{2^n}$ such that $L^*(b)=\gamma^3$ and $\textit{Tr}_n^2(\gamma^{-1}b)\neq0$, where $L^*$ denotes the adjoint linear mapping of $L$.

\item For $n$ odd integer and $L$ linear function over $\mathbb{F}_{2^n}$ the function $x+L(x^3)$ is a permutation of $\mathbb{F}_{2^n}$ if and only if for every $b\in\mathbb{F}_{2^n}$ either $L^*(b)=0$
or $\textit{Tr}_n(\gamma^{-1}b)=0$, where $L^*(b)=\gamma^3$ and $L^*$ denotes the adjoint linear mapping of $L$.
\end{itemize}

The above mentioned function $x^9+\textit{Tr}(x^3)$, for $n=8$, is a clear example of the fact that there are other possible conditions for function of the form  \eqref{1} to be APN.
Therefore with this work we try to find new conditions and new relations for the APN property.

\section{APN Conditions}

\subsection{Necessary and sufficient conditions}
Let $F(x)=L_1(x)+L_2(x^3)$, with $L_1$ and $L_2$ as in (\ref{L}), be a function defined over $\mathbb{F}_{2^n}$ for a positive integer $n$ and $F^\prime(x)=F(x^3)=L_1(x^3)+L_2(x^9)$.\\
Just analysing the APN property for a quadratic function we can state the following lemma.
\begin{lem}\label{1} 
For any positive integer $n$ and any linear functions $L_1$ and $L_2$ of $\mathbb{F}_{2^n}$, a function 
$F^\prime$ defined by (\ref{F'}) is APN if and only if for every $a\in\mathbb{F}_{2^n}^*$ one of the following conditions is satisfied:
\begin{enumerate}
\item   if  $x\neq 0,1$
 \begin{equation} \label{eq:1}
L_1(a^3(x^2+x))+L_2(a^9(x^8+x))\neq 0;
\end{equation}
\item  if   $y\neq 0  \mbox{ and }  \textit{Tr}(y)=0$
 \begin{equation} \label{eq:2}
L_1(a^3y)+L_2(a^9(y^4+y^2+y))\neq 0  .
\end{equation}
\end{enumerate}
\end{lem}

\begin{proof}
Since $F^\prime$ is a quadratic function satisfying $F^\prime(0)=0$, APN condition can be reformulated as the following:\\
  for any $a\in\mathbb{F}_{2^n}^*$
$$F^\prime(ax+a)+F^\prime(ax)+F^\prime(a)=0 \mbox{ if and only if } x\in \{0,1\}.$$
The equation above is equivalent to 
$L_1(a^3(x^2+x))+L_2(a^9(x^8+x))=0$,
   therefore we have that 
$$ L_1(a^3(x^2+x))+L_2(a^9(x^8+x))\neq0\mbox{ if and only if } x\neq 0,1.$$
Lets denote now  $y=x^2+x$. Since $x\neq 0,1$ we have that $y\neq0$ and \textit{Tr}$(y)=0$.
The second condition follows easily.
\end{proof}
\begin{prop}
Let $F^\prime$ be APN and, referring to (\ref{L}), construct the linear function $L_3(x)=\sum_{i+0}^{n-1}d_ix^{2^i}$ with coefficients
$$d_0=b_0+b_{n-1}+c_0+c_{n-3}$$
$$d_1=b_1+b_{0}+c_1+c_{n-2}$$
$$d_2=b_2+b_{1}+c_2+c_{n-1}$$
$$d_i=b_i+b_{i-1}+c_i+c_{i-3},  \mbox{ for } 3\leq i\leq n-1.$$
Then $L_3$ is a 2-to-1 map satisfying $L_3(x)=0$ if and only if $x=0,1$.
\end{prop}

\begin{proof}
Using equation (\ref{eq:1}) with $a=1$, consider the following map: $L_1(x^2+x)+L_2(x^8+x)$.
Analysing the two linear functions we have:
$L_1(x^2+x)=(b_0+b_{n-1})x+\sum_{i=1}^{n-1}(b_i+b_{i-1})x^{2^i}$,
$L_2(x^8+x)=(c_0+c_{n-3})x+(c_1+c_{n-2})x^2+(c_2+c_{n-1})x^{2^2}+\sum_{i=3}^{n-1}(c_i+c_{i-3})x^{2^i}$.\\
Therefore $L_1(x^2+x)+L_2(x^8+x)$ corresponds to the linear function $L_3(x)$  described above.
From Lemma \ref{1} we have that $L_3(x)=0$ if and only if $x=0,1$.
\end{proof}

The following lemma gives a quite fast way to verify if a function $F^\prime$ can be APN, since you have to evaluate it over a third of the elements of the space.

\begin{lem} \label{F'(a)ne0}
For $n$ even assume $F^\prime$ is APN. Let $\alpha\in\mathbb{F}_{2^n}^*$ be a primitive element and $k=\frac{2^n-1}{3}$.
 Then $F^\prime(a)\neq0$ for any $a\neq0$ or equivalently  $F(\alpha^{3j})=F^\prime(\alpha^j)\neq0$ for $0\leq j\leq k-1$.
\end{lem}
\begin{proof}
For $n$ even we have \textit{Tr}$(1)=0$. Therefore using equation (\ref{eq:2}) with $y=1$ we get for any $a\neq0$
$$ L_1(a^3)+L_2(a^9)=F(a^3)=F^\prime(a)\neq0.$$
For $a\neq0$ we have that $a=\alpha^j$ with $0\leq j\leq 2^n-2$.
Since we consider just cubic power of $a$, we can restrict the possibilities to $0\leq j\leq k-1$. This concludes the proof.
\end{proof}

\begin{oss}
If we consider $j=0$ in Lemma \ref{F'(a)ne0} then $$L_1(1)+L_2(1)=\sum_{i=0}^{n-1}b_i+\sum_{i=0}^{n-1}c_i=\sum_{i=0}^{n-1}(b_i+c_i)\neq0.$$
Moreover, if we just consider linear functions defined over $\mathbb{F}_2$ (i.e. $b_i,c_i\in\mathbb{F}_{2}$) then a fast way to check if $F^\prime$ is not APN is by verifying that $L_1$ and $L_2$ have the same parity number of monomials.
\end{oss}

\begin{lem}
Let $n$ be an even number multiple of 3 and $F^\prime$ be APN. Then for any $a\neq0$ $L_1(a^3\beta)\neq 0$, with $\beta\in\mathbb{F}_{2^3}^*$ such that $\textit{Tr}_3(\beta)=0$.
\end{lem}
\begin{proof}
Consider such an element $\beta$ and call $m$ the integer $\frac{n}{3}$.
We have that $\textit{Tr}_n(\beta)$ is equal to $\sum_{j=1}^m\sum_{i=0}^2\beta^{2^i}=\sum_{j=1}^m\textit{Tr}_3(\beta)=0$.
Therefore we can apply (\ref{eq:2}) with $y=\beta$ and obtain
$$L_1(a^3\beta)+L_2(a^9(\beta^4+\beta^2+\beta))=L_1(a^3\beta)\neq 0 \ \forall a\neq0.$$
\end{proof}

\begin{lem} \label{APNiif}
Consider a function $F^\prime$ from $\mathbb{F}_{2^n}$ to itself defined as in \eqref{F'}.
$F^\prime$ is APN if and only if it satisfies the following condition:

for every $a,y\neq0$  with \textit{Tr}($y$)=0 if an element $t\in\mathbb{F}_{2^n}$ satisfies $\textit{Tr} (t$)=0 and
$$L_1(a^3y)=L_2(a^9y^3t)$$ then  $L_2(a^9(y^4+ty^3+y^2+y))\neq0$.

\end{lem}

\begin{proof}
By Lemma \ref{1} we have that APN property for $F^\prime$ is equivalent to
$$\mbox{for any } a,y\in\mathbb{F}_{2^n}^*, \textit{Tr}(y)=0 \ L_1(a^3y)+L_2(a^9(y^4+y^2+y))\neq0.$$
Assume that there exists an element $t$ that satisfies the conditions in the statement.
Let us re-write the formula above as\\
$\begin{array} {cl} 
0\neq & L_1(a^3y)+L_2(a^9(y^4+y^2+y))=\\
&L_1(a^3y)+L_2(a^9y^3t)+L_2(a^9(y^4+ty^3+y^2+y))=\\
&L_2(a^9(y^4+ty^3+y^2+y)).
\end{array}$

Therefore the APN condition is equivalent to $$L_2(a^9(y^4+ty^3+y^2+y))\neq0.$$
On the other hand assume that for any $t$ of null trace we have  $L_1(a^3y)\neq L_2(a^9y^3t)$.
Therefore \begin{center}
$L_1(a^3y)\not\in \Omega=\{ L_2(a^9y^3t) : \textit{Tr}(t)=0 \}$.
\end{center}
Let us consider the second term of the formula,\begin{center}
 $L_2(a^9(y^4+y^2+y))=L_2(a^9y^3(y+1/y+1/y^2))$. 
\end{center}
Since $\textit{Tr}(y+1/y+1/y^2)$=0, the term belongs to the set $\Omega$.
Therefore the relation is again respected.
\end{proof}

\begin{cor}
For general $a\neq0$ and $y\neq0$ with \textit{Tr}($y$)=0, if the equation $$L_1(a^3y)=L_2(a^9y^3t)$$
is satisfied only for $t$ with \textit{Tr}($t$)=1, then the function $F^\prime(x)=F(x^3)=L_1(x^3)+L_2(x^9)$ is APN. 
\end{cor}
\begin{proof}
In this case the hypothesis of the previous lemma is always satisfied, since there is no element $t$ such that $L_1(a^3y)=L_2(a^9y^3t)$ and $\textit{Tr}(t)=0$.
Therefore the function $F^\prime$ is APN.
\end{proof}

\subsection{On APN functions of the form $x^9+L(x^3)$}
From \cite{BCL09} we know that in $\mathbb{F}_{2^8}$ the function $F^\prime(x)=x^9+\textit{Tr}(x^3)$ is APN.
\begin{lem}
If $3|n$ then the function $x^9+\textit{Tr}(x^3)$ is not  APN over $\mathbb{F}_{2^n}$.
\end{lem}
\begin{proof}
From Lemma \ref{1} we have that $x^9+\textit{Tr}(x^3)$ is APN if and only if for any $a\neq0$ and any $x\neq0,1$
$$\textit{Tr}(a^3(x^2+x))+a^9(x^8+x)\neq0.$$
If we now consider $n$ multiple of 3, $x\in\mathbb{F}_{2^3}\setminus\mathbb{F}_2$ and $a=1$ we obtain
$$\textit{Tr}(a^3(x^2+x))+a^9(x^8+x)=0.$$
\end{proof}

Using Lemma \ref{APNiif} it was possible to implement, using the software MAGMA, a fast algorithm that checks if $x^9+\textit{Tr}(x^3)$ is APN over $\mathbb{F}_{2^n}$.
Running the code for $n$ up to 200  the only APN functions are for dimensions 4, 5 and 8.

Let us consider now a more general form for $F^\prime$, $G(x)=x^9+L(x^3)$ with $L$ linear function in $\mathbb{F}_{2^n}[x]$.

With some computational work, done with MAGMA, we tried to find more APN functions of this form in other dimensions. 
In Table \ref{tab1} we summarize the results we obtained.
With $\alpha$ we indicate a primitive element of $\mathbb{F}_{2^n}^*$.
 We searched for APN functions in $\mathbb{F}_{2^n}$, up to $n=10$, of the form $x^9+L(x^3)$.
We studied their CCZ-equivalence relation and obtained the representatives for each dimension.
Obviously for every $n$ not multiple of 3, the Gold function $x^9$ (corresponding to the case $L(x)=0$) is APN.

\begin{table}[h]
\centering
\label{tab1}
\caption{APN functions of the form $x^9+L(x^3)$ over $\mathbb{F}_{2^n}$}\label{tab1}\vspace{4pt}
\renewcommand{\arraystretch}{1.4}
\begin{tabular}{c||  c|  c|} 
$n$ &  CCZ-eq & Representatives\\
\hline\hline
4 &  1  & $L(x)=0$\\\hline
5 &  2 & $L(x)=0$, $L(x)=\textit{Tr}(x)$\\\hline
6 &  2  & $L(x)=\alpha^{44}x+\alpha x^2$, \\
&& $L(x)=\alpha^{23}x+x^{2^2}$\\\hline
7  & 1 & $L(x)=0$\\ \hline 
8   & 6 & $L(x)=0$, $L(x)=x^2+x^{2^4}$,\\
    &   & $L(x)=x^{2^3}+x^{2^7}$, $L(x)=\textit{Tr}(x)$, \\
    &  & $L(x)=x^{2^2}+\alpha^{85}x^{2^3}+x^{2^4}$, \\
    & & $L(x)=\alpha^{60}x+ \alpha^{200}x^2+ \alpha^{242}x^4+ \alpha^{190}x^8   + \alpha x^{16}   $\\ \hline
9  & 0  & -\\\hline
10  & 2 & $L(x)=0$,\\
&& $L(x)=\alpha^{1021}x+\alpha^{1022}x^2+\alpha x^{2^2} $\\\hline

\end{tabular}
\end{table}

For greater dimensions we just checked the possible APN function of the form $x^9+L(x^3)$ with $L\in\mathbb{F}_2[x]$, up to CCZ-equivalence.
\begin{itemize}
\item for $n=11$ there are no APN except $F(x)=x^9$;
\item for $n=12$ there are no APN;
\item for $n=13$ there are no APN except $F(x)=x^9$.
\item for $n=14$ there are no APN except $F(x)=x^9$;
\item for $n=15$ there are no APN;
\item for $n=16$ there are no APN except $F(x)=x^9$.
\end{itemize} 

For $n=4$ the function $x^9+\textit{Tr}(x^3)$ is CCZ-equivalent to the Gold function $x^9$.\\
For $n=6$ the found APN functions are not CCZ-equivalent to functions $x^9+L(x^3)$ defined over $\mathbb{F}_2$.
Moreover we get:
\begin{itemize}
\item  for $L(x)=\alpha^{44}x+\alpha x^2$ the function $x^9+L(x^3)$  is CCZ-equivalent to $x^3+\alpha^{-1}\textit{Tr}_n(\alpha^3x^9)$;
\item for $L(x)=\alpha^{23}x+x^{2^2}$ the function $x^9+L(x^3)$ is CCZ-equivalent to $x^3=x^3+\textit{Tr}_n(x^9)$.
\end{itemize}
Both of these functions belong to the class of APN functions studied in \cite{BCL09.2}\\
For $n=8$ we compared the found APN mappings with the list of known APN function in dimension 8 in \cite{EP08}.
We get the following:
\begin{itemize}
\item for $L(x)=x^2+x^{2^4}$ the function $x^9+L(x^3)$ is CCZ-equivalent to $x^3+\textit{Tr}(x^9)$;
\item for $L(x)=x^{2^3}+x^{2^7}$ the function $x^9+L(x^3)$ is CCZ-equivalent to $x^3$;
\item for $L(x)=x^{2^2}+\alpha^{85}x^{2^3}+x^{2^4}$ the function $x^9+L(x^3)$ is not CCZ-equivalent to any function of the form $x^3+a^{-1}\textit{Tr}(a^3x^9)$ but it is CCZ-equivalent to function $\alpha^{135}x^{144} + \alpha^{120}x^{66} + \alpha^{65}x^{18} + x^3 $, no. 6 in the list of APN mapping of $\mathbb{F}_{2^8}$ in \cite{EP08};
\item for $L(x)=\alpha^{60}x+ \alpha^{200}x^2+ \alpha^{242}x^4+ \alpha^{190}x^8   + \alpha x^{16} $ the function $x^9+L(x^3)$ is not CCZ-equivalent to any function of the form $x^3+a^{-1}\textit{Tr}(a^3x^9)$ but it
 is CCZ-equivalent to function $\alpha^{242}x^{192} + \alpha^{100}x^{144} + \alpha^{66}x^{132} + \alpha^{230}x^{129} + \alpha^{202}x^{96} + \alpha^{156}x^{72} + \alpha^{254}x^{66} + \alpha^{18}x^{48} 
 + \alpha^{44}x^{36} + \alpha^{95}x^{33} + \alpha^{100}x^{24}  + \alpha^{245}x^{18} + \alpha^{174}x^{12} + \alpha^{175}x^9   + \alpha^{247}x^6 +    \alpha^{166}x^3$, 
   no. 9 in the list of APN mapping of $\mathbb{F}_{2^8}$ in \cite{EP08}

\end{itemize}

\subsection{On the number of bent components}
From \cite{BCCL06} we get the following theorem.
\begin{theorem}
Let $F$ be a function from $\mathbb{F}_{2^n}$ to $\mathbb{F}_{2^n}$.
Then for any non-zero $a\in\mathbb{F}_{2^n}$ $$\sum_{\lambda\in\mathbb{F}_{2^n}}\mathcal{F}^2(D_af_\lambda)\geq 2^{2n+1}.$$
Moreover, $F$ is APN if and only if for every non-zero $a\in\mathbb{F}_{2^n}$ $$\sum_{\lambda\in\mathbb{F}_{2^n}}\mathcal{F}^2(D_af_\lambda)=2^{2n+1}.$$   
\end{theorem}

\begin{lem} \label{lemma}
$F^\prime(x) = L_1(x^3)+L_2(x^9)$ is an APN function if and only if for any $a\in\mathbb{F}_{2^n}^*$ there exists one and only one $\lambda\in\mathbb{F}_{2^n}^*$ such that 
$\textit{Tr}(\lambda L_1(ax^2+a^2x)+\lambda L_2(ax^8+a^8x))$ is constantly 0.
\end{lem}
\begin{proof}
Since $F^\prime$ is a quadratic function, every component has at most  algebraic degree 2 and, consequently, the Boolean function $D_a f^\prime_\lambda$ can be either affine or constant.
If $D_af^\prime_\lambda$ is affine then $\mathcal{F}(D_af^\prime_\lambda)=0$.
In the other case we have $\mathcal{F}(D_af^\prime_\lambda)=\pm2^n$ and $\mathcal{F}^2(D_af_\lambda)=2^{2n}$.
Let consider  the set
\begin{equation}\label{Delta}
\Delta_a=\{ \lambda\in\mathbb{F}_{2^n} :\ D_a f^\prime_\lambda \mbox{ is constant}\},
\end{equation}
  then
$$\sum_{\lambda\in\mathbb{F}_{2^n}}\mathcal{F}^2(D_af^\prime_\lambda) = 2^{2n}\cdot|\Delta_a|.$$
From the previous theorem we have that $F'$ is APN if and only if the sum is equal to $2^{2n+1}$, hence if and only if $|\Delta_a|=2$.
Since $f^\prime_0$, and consequently $D_a f^\prime_0$, is the constantly null function, we have that 0 belongs to the set $\Delta_a$ for every $a\neq0$.
Therefore, $F^\prime$ is APN if and only if $|\Delta^*_a|=1$, with
\begin{equation} \label{Delta*}
\Delta^*_a = \Delta_a \setminus \{0\}.
\end{equation}
This is true for every generic quadratic APN function $F^\prime(x)$.\\
In our specific case we have

$\begin{array}{ll}
 D_af^\prime_\lambda(x)&= \textit{Tr}(\lambda[F'(x)+F'(x+a)])=\\
& \textit{Tr}(\lambda[L_1(ax^2+a^2x+a^3)+\\
&\qquad +L_2(ax^8+a^8x+a^9)])=\\
&  \textit{Tr}(\lambda[L_1(ax^2+a^2x)+L_2(ax^8+a^8x)])+\\
&\qquad+\textit{Tr}(\lambda[L_1(a^3)+L_2(a^9)]).
\end{array}$

 In order to study its constant conditions it is sufficient to study the function $g(x)=\textit{Tr}(\lambda[L_1(ax^2+a^2x)+L_2(ax^8+a^8x)])$.
Since $g(0)=0$, we have that if $g$ is constant then it is the constant zero function and this concludes the proof.
\end{proof}

\begin{oss}
Equivalently, we can study the conditions for $\textit{Tr}(\lambda L_1(a^3[x^2+x])+\lambda L_2(a^9[x^8+x]))$ to be the constant null function.
Due to the property of the Trace function we can study the conditions for $\lambda L_1(a^3[x^2+x])+\lambda L_2(a^9[x^8+x])$ to be equal to $\eta+\eta^2$, with $\eta=\eta(a,\lambda,x)$.
\end{oss}

Recalling the notation used in the proof we defined:\\
$\Delta_a$ as in \eqref{Delta} and  $\Delta_a^*$ as in \eqref{Delta*};
\begin{equation}\label{Vl}
V_\lambda = \{ a\in\mathbb{F}_{2^n} : D_a f^\prime_\lambda \mbox{ is constant}\};
\end{equation}
\begin{equation} \label{Vl*}
V_\lambda^* = V_\lambda \setminus \{0\}.
\end{equation}

From Proposition 1 in \cite{CCK08} we get that the dimension of the kernel of $f_\lambda$ and $n$ have the same parity, 
where the kernel of a quadratic form $f$ is the subspace of $\mathbb{F}_{2^n}$ \\
$\{u\in\mathbb{F}_{2^n} : f(u+v)+f(u)+f(v)=0 \ \mbox{for any }v\in\mathbb{F}_{2^n}\}$.
From Lemma 1 in \cite{CCK08} we get that, since $f_\lambda$ is a quadratic Boolean function, its kernel corresponds to the subspace $V_\lambda$.
Therefore we have dim$_{\mathbb{F}_2}(V_\lambda) \equiv$ $n$ mod 2.

Moreover let's consider the set
\begin{equation*}
\Gamma_i=\{ \lambda\neq 0 : \mbox{dim}(V_\lambda)=i\}.
\end{equation*}
If $\Gamma_i$ not empty then $i$ has the same parity as $n$.
It can be easily proved by considering  a not null element $\lambda$ in the set, i.e. such that dim$V_\lambda=i$.
Since the dimension of $V_\lambda$ has the same parity as $n$, the same can be state on $i$. \\
The set $\Gamma_0$ correspond to the set of all bent components.

\begin{cor}
From  Lemma \ref{lemma} it is straightforward to prove that APN property for a quadratic function is equivalent to the following statement:
$\mbox{for\ any}\ \lambda_1\neq \lambda_2 \in\mathbb{F}^*_{2^n},$ $$\ V_{\lambda_1}\cap V_{\lambda_2} = \emptyset\ \mbox{ and }\ \sum_{\lambda\neq0}|V_\lambda^*|=2^n-1.$$
\end{cor}

\subsubsection{Computational Results}
Using the software MAGMA we tried to verify for functions $F^\prime$ of form \eqref{F'} defined over small dimensions the relation between the APN property and the number of bent components.
From the results obtained taking random linear functions $L_1, L_2$ and constructing $F^\prime$ for $n\in\{4,6,8\}$ the relation seems the following one:
\begin{con}
For an even $n$, a function $F^\prime$ of the form \eqref{F'} is APN if and only if it has exactly $\frac{2}{3}(2^n-1)$ bent components.
\end{con} 

We know that this is not true for general quadratic functions.
Indeed consider the quadratic APN function presented by Dillon in 2006 \cite{Dillon06}
$$
F(x)=x^3+u^{11}x^5+u^{13}x^9+x^{17}+u^{11}x^{33}+x^{48};
$$\\
defined over $\mathbb{F}_{2^6}$ where $u$ is a primitive element, root of the polynomial $x^6+x^4+x^3+x+1$.
This function has 46 bent components and $46>\frac{2}{3}(2^6-1)=42$.

\section{Conclusion}
In this work we continued the study of quadratic functions of the form $L_1(x^3)+L_2(x^9)$ and related APN conditions.
New necessary and sufficient conditions are presented in this paper.
Such conditions allow us to compute a faster algorithm that checks the existence of other APN functions of such form.
New results are given considering functions of the form $x^9+L(x^3)$.
Up to CCZ-equivalence new APN functions are found in different low dimensions.

\end{document}